\documentclass[a4paper]{article}
\usepackage{algorithmic,algorithm}
\usepackage{amsmath, amssymb, amsthm}
\usepackage{multirow}
\usepackage{subfigure}
\usepackage{epsfig} 
\usepackage{url}

\title{On the Existence of Hamiltonian Paths for History Based Pivot Rules on Acyclic Unique Sink Orientations of Hypercubes}
\author{
Yoshikazu Aoshima\thanks{Department of Computer Science
Graduate School of Information Science and Technology,
The University of Tokyo / ERATO-SORST Quantum Computation and
Information Project, JST.
\texttt{y-aoshima@is.s.u-tokyo.ac.jp}}
\and
David Avis\thanks{School of Informatics, Kyoto University
and School of Computer Science, McGill University \texttt{avis@cs.mcgill.ca}}
\and
Theresa Deering\thanks{School of Computer Science, McGill University
\texttt{theresa.deering@mail.mcgill.ca}}
\and
Yoshitake Matsumoto\thanks{Department of Computer Science
Graduate School of Information Science and Technology,
The University of Tokyo (until 2009)/ Google Japan Inc. \texttt{ymatsu@is.s.u-tokyo.ac.jp}}
\and
Sonoko Moriyama\thanks{Graduate School of Information Sciences, Tohoku University \texttt{moriso@dais.is.tohoku.ac.jp}}
}

\begin{document}
\setlength{\abovedisplayskip}{1pt}   \setlength{\belowdisplayskip}{1pt}   \theoremstyle{definition}
\newtheorem{thm}{Theorem}[section]
\newtheorem{lem}{Lemma}[section]
\newtheorem{dfn}{Definition}[section]
\newtheorem{col}{Corollary}[section]

\makeatletter
\newcommand{\figcaption}[1]{\def\@captype{figure}\caption{#1}}
\newcommand{\tblcaption}[1]{\def\@captype{table}\caption{#1}}
\makeatother

\maketitle
\begin{abstract}
An acyclic USO on a hypercube is formed by directing its edges in such as way 
that the digraph is acyclic and each face of the hypercube has a unique sink and a unique source. 
A path to the global sink of an acyclic USO can be modeled as pivoting in a unit 
hypercube of the same dimension with an abstract objective function, and vice versa. 
In such a way, Zadeh's 'least entered rule' and other history based pivot rules can be 
applied to the problem of finding the global sink of an acyclic USO. In this 
paper we present some theoretical and empirical results on the existence of acyclic USOs
for which the various history based pivot rules 
can be made to follow a
Hamiltonian
path. 
In particular, we develop an algorithm that can enumerate all such paths up to dimension 6 using 
efficient pruning techniques. We show that Zadeh's original rule
admits Hamiltonian paths up to dimension 9 at least, and prove that most of the other rules do not for all
dimensions greater than 5.
\end{abstract}

\section{Introduction}

It is now over 30 years since Khachian showed that linear programming
problems can be solved in polynomial time \cite{khachiian1980polynomial}.
His ellipsoid algorithm and subsequent interior point methods are not,
however, strongly polynomial time algorithms and no such
algorithms are known. Pivoting algorithms, such as Dantzig's simplex
method \cite{dantzig2003linear} still offer the possibility
of being strongly polynomial. One reason for this is that pivoting algorithms
follow a path on the graph defined by
the skeleton of a polyhedron, and it is widely believed
that the diameter of this graph is polynomially bounded in the
dimensions of the linear program.

In fact, Hirsch conjectured that the diameter of
any $d$-dimensional polytope with $n$ facets, where $n > d \geq 2$, 
is less than or equal to $n-d$.
Very recently Santos found that this conjecture is false, by exhibiting
a polytope with dimension $d=43$ and $n=86$ facets 
with diameter equal to $n-d+1$ \cite{santos2010counterexample}. 
Nevertheless, the belief that the diameter is polynomial is still strong.
The subexponential bounds of Kalai \cite{kalai1992subexponential}
and Matou{\v{s}}ek, Sharir, Welzl \cite{MSW92}
also give further grounds for hope.
These papers use randomized pivot selection rules, and no deterministic rules
that achieve these subexponential bounds are known.

The simplex method is a family of algorithms, with each member of the family
being determined by a pivot selection rule.
In practice, Dantzig's original rule works extremely well.
However, Klee and Minty \cite{klee1972good} constructed 
a case where the simplex method 
using this rule follows an exponential length path on a family
of suitably stretched hypercubes, since called 
{\em Klee-Minty cubes}. 
In fact it visits every vertex, that is, Hamiltonian path, of the hypercube 
(see Figure \ref{fig:klee-minty}).
\begin{figure}[ht]
\begin{center}
\subfigure[Dimension 1 or 2]{
\includegraphics[height=0.1\textheight  ]{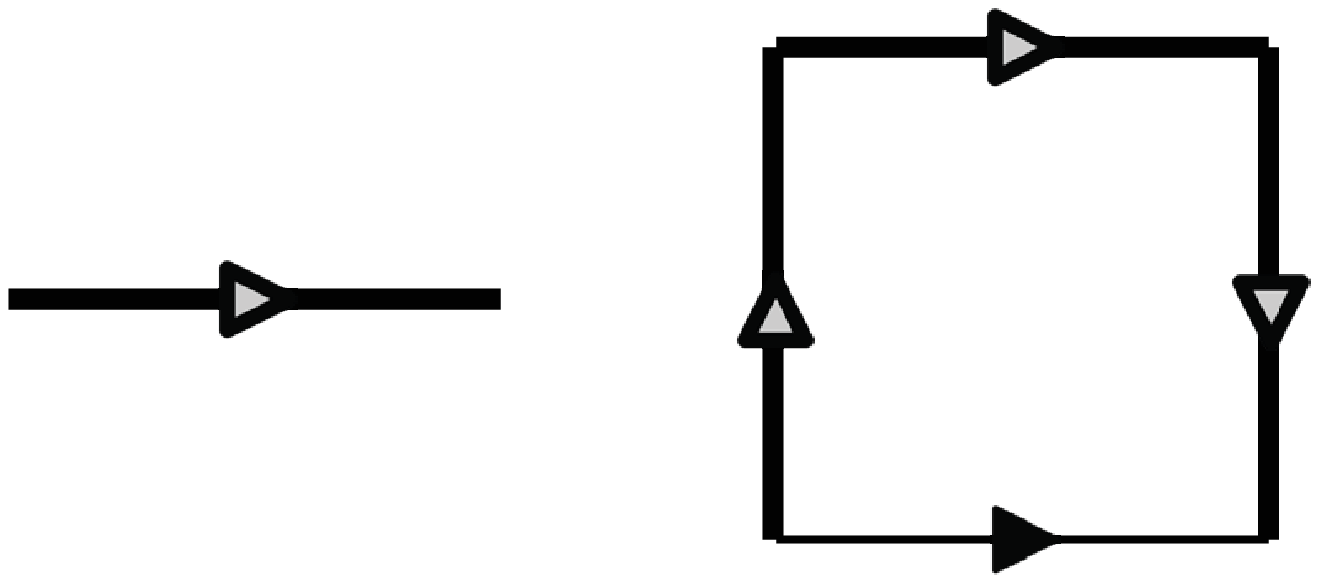}
\label{fig:1or2dimensionalKM}
}
\subfigure[Dimension 3]{
\includegraphics[height=0.1\textheight  ]{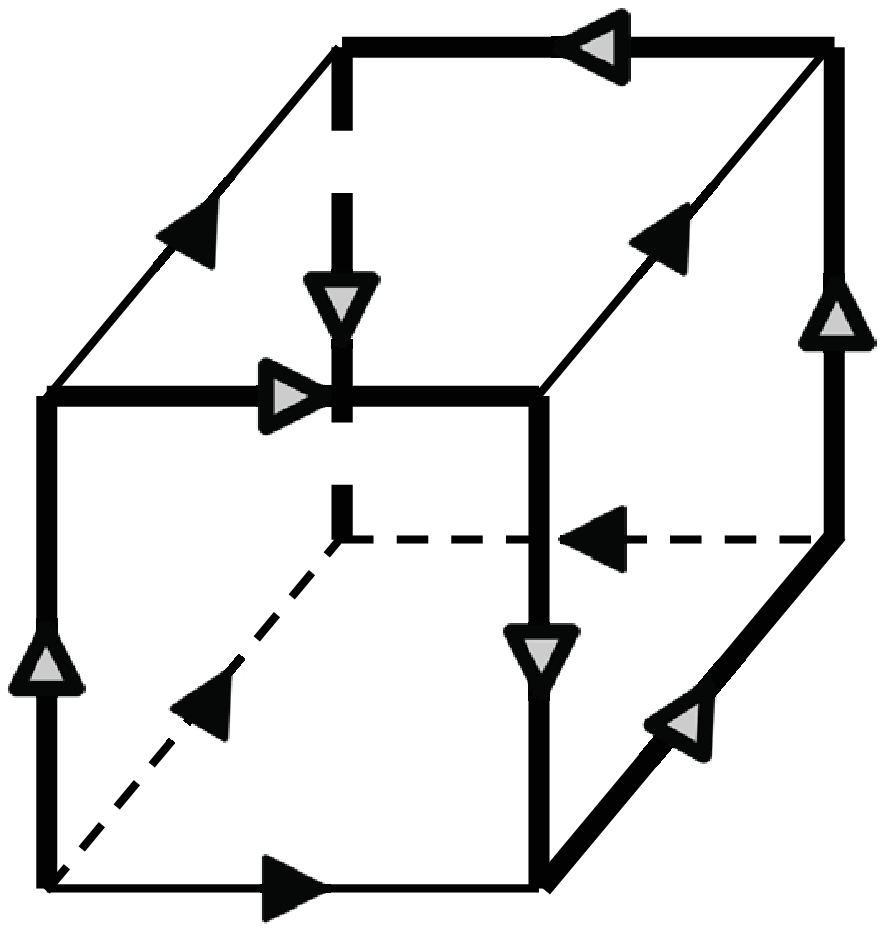}
\label{fig:3dimensionalKM}
}
\subfigure[Dimension 4]{
\includegraphics[height=0.2\textheight  ]{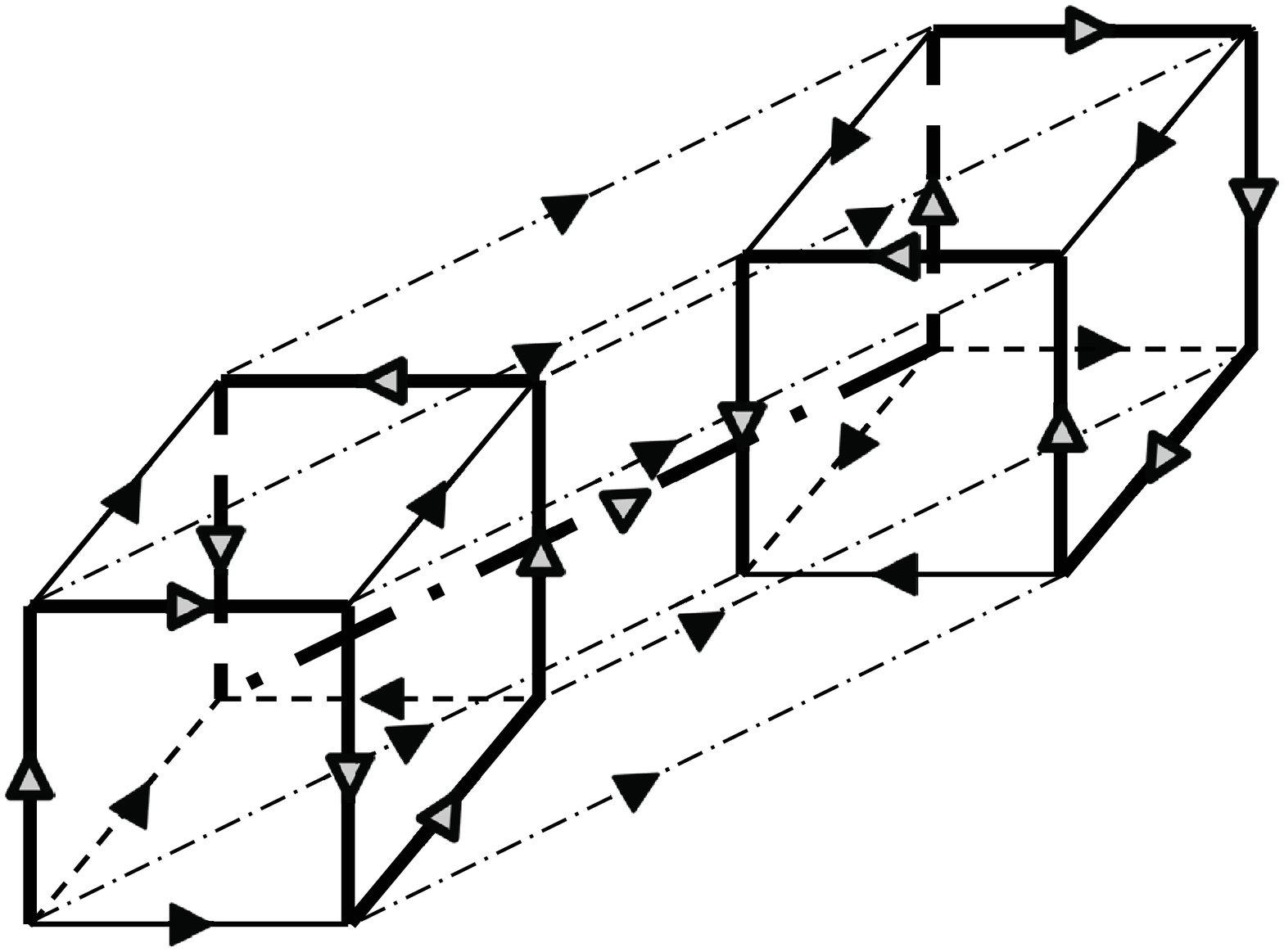}
\label{fig:4dimensionalKM}
}
\end{center}
\caption{Klee-Minty cube. 
}
\label{fig:klee-minty}
\end{figure}
Subsequent research demonstrated that many other pivot rules take 
exponential time on variants of the Klee-Minty examples.
Such pivot rules include the maximum improvement rule 
(Jeroslow \cite{jeroslow1973simplex}) and Bland's rule 
(Avis and Chv\'atal \cite{avis1978notes}).

There are still some pivot rules for which the behaviour of the simplex method is unknown.
A particularly interesting set of these rules are the history based pivot rules, of which the best known
is the {\em least entered rule} proposed by Zadeh in a 1980 Stanford University
Technical Report, that was recently reprinted \cite{zadeh2009worst}. 
Very recently, 30 years after it originally appeared, Friedmann showed that this rule requires 
at least sub-exponential time in the worst case such as $2^{\Omega(\sqrt{d})}$ \cite{Fr10}. A non-trivial
upper bound on Zadeh's rule is still unknown.

To motivate the least entered rule, Zadeh pointed out a characteristic of Klee-Minty examples: 
some variables pivot very few times
and other variables pivot an exponential number of times.
Zadeh's pivot rule avoids this behaviour by making each variable pivot, roughly, the same
number of times. In this way, it behaves similarly to the random pivot selection rules
mentioned above. In Section \ref{section:searching} 
we show that if Zadeh's rule follows a 
Hamiltonian path on a hypercube, then indeed, each variable must pivot an
exponential number of times in the dimension $d$ of the cube.

Zadeh's rule differs from the former pivot rules in that it
uses information from the entire pivot history up to that point.
Such pivot rules are called {\em history based pivot rules}. 
Besides Zadeh's rule, these include the least-recently basic rule 
\cite{cunningham1979theoretical}, the least-recently considered rule 
\cite{cunningham1979theoretical}, the least-recently entered rule 
\cite{fathi1986affirmative}, and the least iterations in the basis rule \cite{avis2009history}.
We remark that for each of these pivot rules, there exists no known exponential lower bound.

In this paper, we study the behaviour of history based pivot rules on an abstraction
of linear programming known as acyclic unique sink orientations (AUSOs) of hypercubes,
that were introduced by Szab\'{o} and Welzl \cite{szabo2001unique}.
These are orientations of the hypercube so that the resulting directed graph is
acyclic, and each face of each dimension has a unique sink and a unique source, 
see Figure \ref{fig:AUSO}.
The goal is to find the unique sink of the hypercube. 

\begin{figure}[!ht]
\begin{center}
\subfigure[]{
\includegraphics[height=0.1\textheight ]{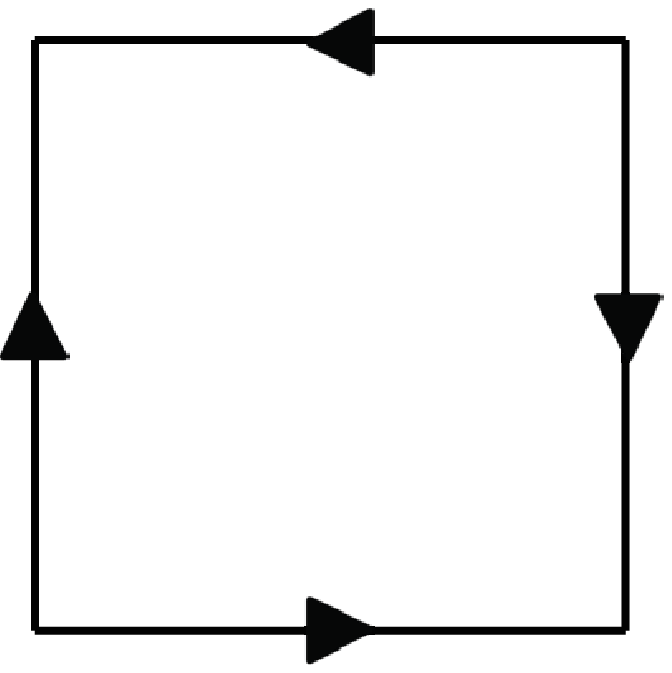}
\label{fig:2dimensionalnonUSO}
}
\subfigure[]{
\includegraphics[height=0.1\textheight  ]{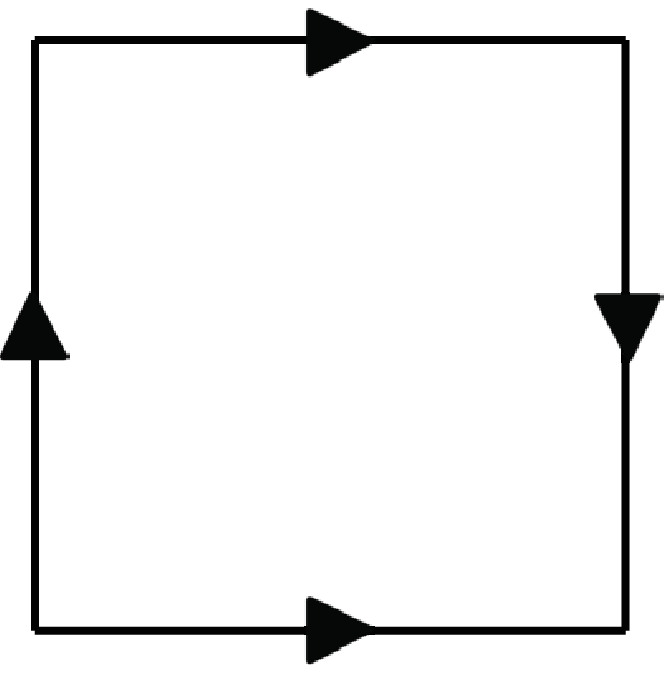}
\label{fig:2dimensionalUSO}
}
\subfigure[]{
\includegraphics[height=0.1\textheight ]{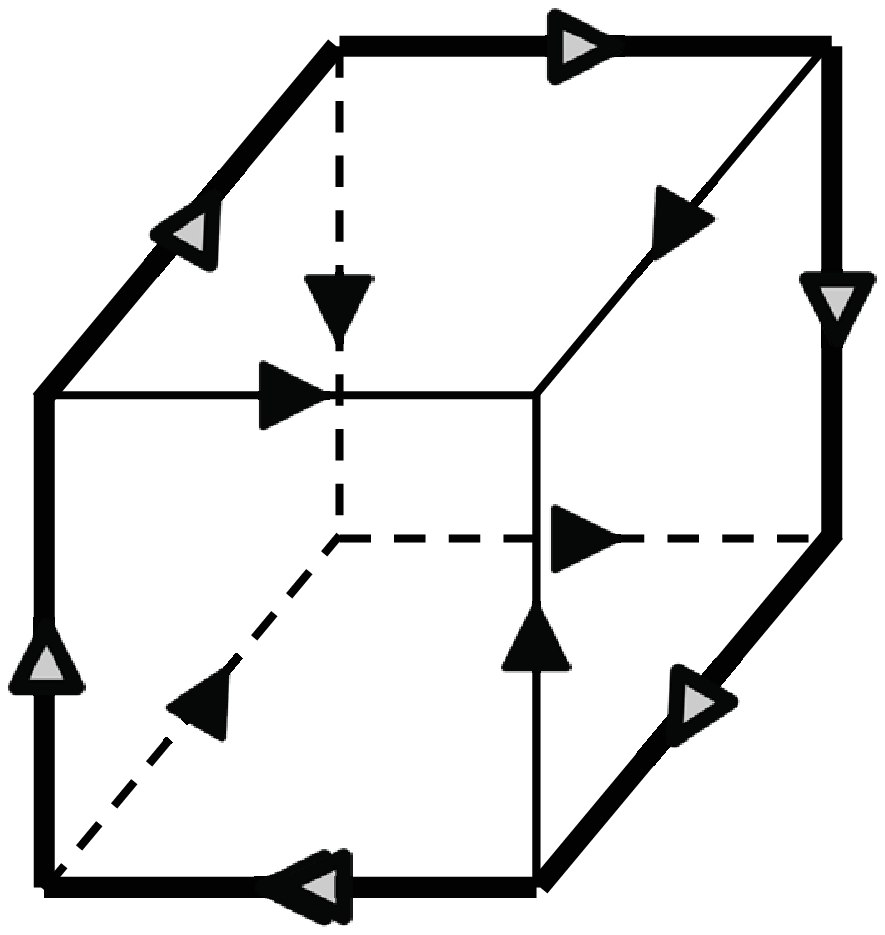}
\label{fig:3dimensionalUSO}
}
\subfigure[]{
\includegraphics[height=0.1\textheight  ]{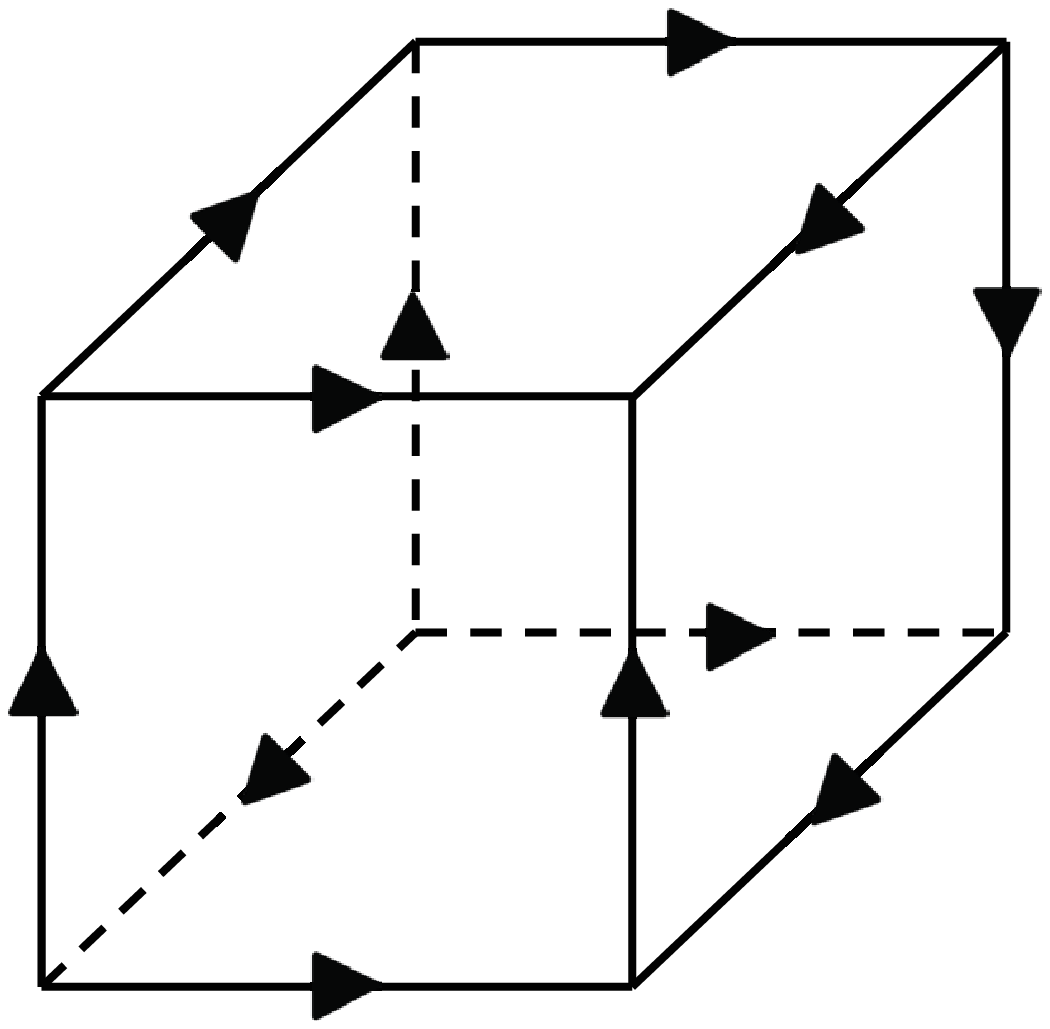}
\label{fig:3dimensionalAUSO}
}
\end{center}
\caption{ \ref{fig:2dimensionalnonUSO} is a non-USO, 
\ref{fig:2dimensionalUSO} and \ref{fig:3dimensionalAUSO} are AUSO cubes, 
\ref{fig:3dimensionalUSO} is a USO cube which has a cycle. }
\label{fig:AUSO}
\end{figure}

As noted in \cite{szabo2001unique},
various optimization problems can be solved using this model. The direction
on an edge of the hypercube corresponds to an increase in the value of
an abstract objective function defined on the vertices of the hypercube.
The concept of abstract objective functions was first introduced by Adler and Saigal
\cite{AS76}. The related concept of completely unimodal numberings was introduced by
Williamson Hoke \cite{hoke1988completely}. 
Although AUSOs need not correspond to actual polytopes and objective
functions, the notions of linear programming, such as bases and pivots, 
are readily available. Therefore we obtain an abstract model
on which to observe the behaviour of various history based pivot rules.

AUSOs on $d$-dimensional
hypercubes have a structure that makes for convenient notation and terminology.
Each vertex is labelled $0, ..., (2^d - 1)$ such that the binary 
representation of adjacent vertices' labels differ by exactly one bit.
Each edge has a \emph{direction} and an \emph{orientation}.
The direction is given by a number $1, ..., d$ indicating which 
bit is different between the two endpoints (counted right-to-left).
The orientation is given by a positive sign (+) if the differing 
bit is 0 at the edge's tail and 1 at its head, and it is 
given by a negative sign (-) otherwise.
We will use the terms {\em direction} to denote which bit is to be changed and
{\em signed direction} which also specifies the orientation.
For emphasis,
to specify the direction without sign we use the term {\em unsigned direction} (see Figure \ref{fig:dir}).
\begin{figure}[!ht]
\begin{center}
\subfigure[]{
\includegraphics[height=0.2\textheight ]{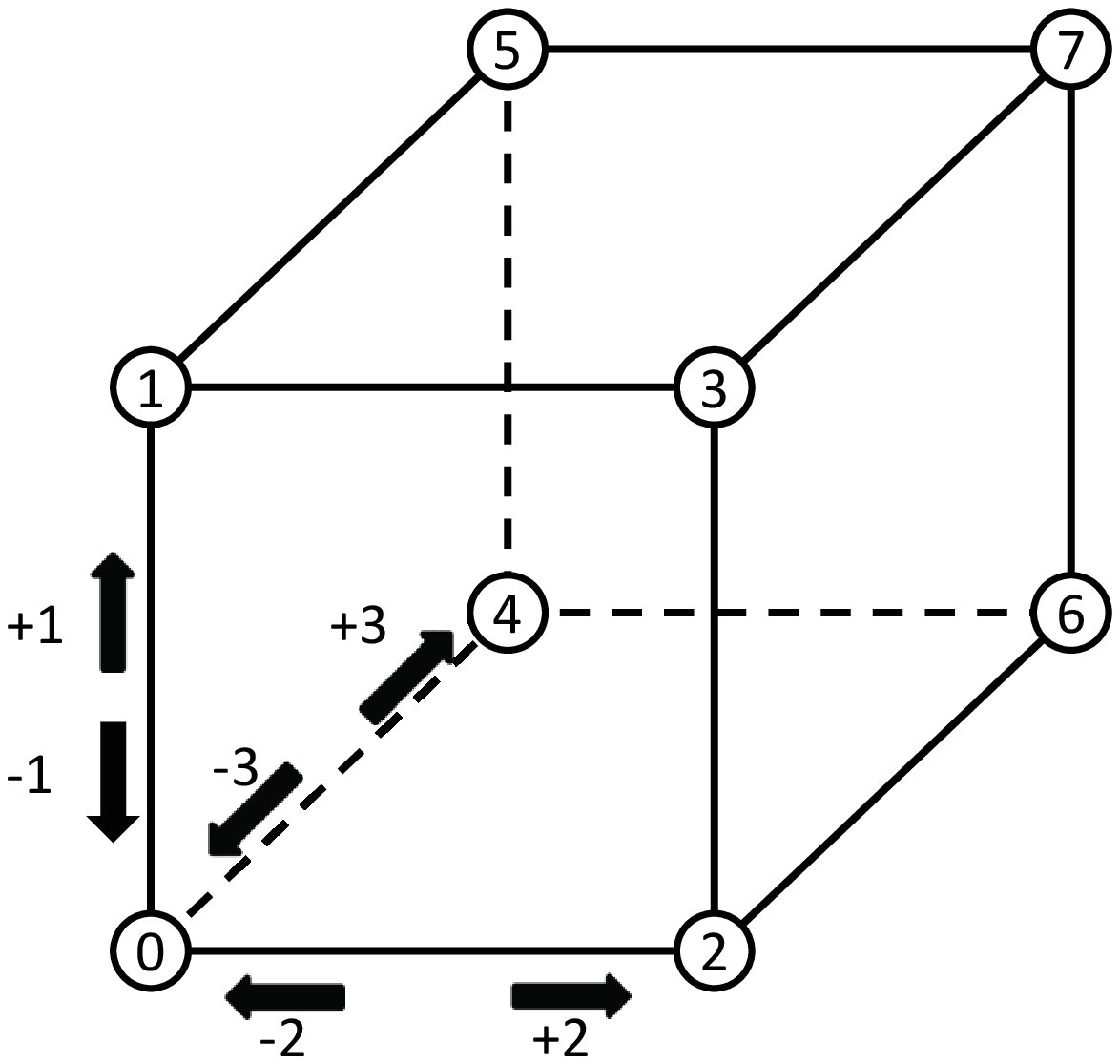}
\label{fig:signeddir}
}
\subfigure[]{
\includegraphics[height=0.2\textheight  ]{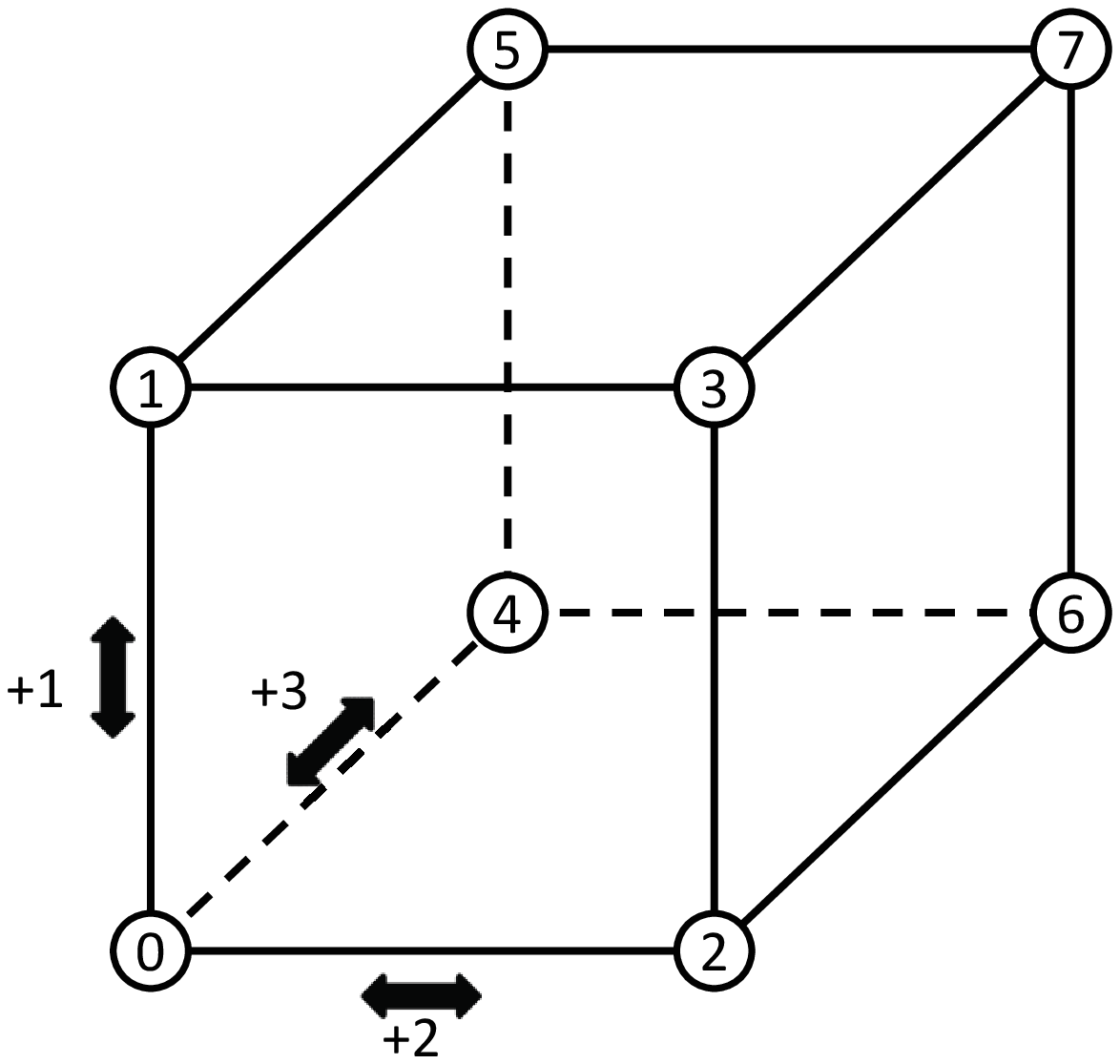}
\label{fig:usigneddir}
}
\end{center}
\caption{Signed and unsigned direction on a 3-dimensional cube}
\label{fig:dir}
\end{figure}
For some pivot rules only the unsigned direction is important 
whereas for others both the orientation and direction are important.

Although AUSOs do not necessarily correspond to LP digraphs, 
the above vertex labelling can be used to model moving 
along a path on an AUSO as pivoting in the dictionary 
$x_{d+i} = 1 - x_i$ for $i = 1, ..., d$.
A pivot $+i$ corresponds to a pivot where $x_i$ enters 
the basis and $x_{d+i}$ leaves, and a pivot $-i$ 
corresponds to a pivot where $x_{d+i}$ enters the basis and $x_i$ leaves.
This allows the AUSO to inherit various pivoting strategies 
that are defined in terms of LPs.

The Klee-Minty examples mentioned above can be modeled as AUSOs.
In fact, they show that Dantzig's original pivot rule for the simplex method leads to
a Hamiltonian path on an associated AUSO for each dimension $d$. 
As mentioned, similar results have been found for
other pivot rules. In this paper, we investigate whether history based pivot rules
can lead to Hamiltonian paths on AUSOs.

Our focus on Hamiltonian paths has the following motivations. Firstly, if such paths exist
for a given pivot rule they are obviously the worst case examples. Secondly,
the number of AUSOs is extremely large.
Stickney showed their are 19 in 3 dimensions \cite{stickney1978digraph}, 
Moriyama's program showed there are 12640 in 4 dimensions \cite{moriyama2006enumeration}, 
and Matou{\v{s}}ek \cite{Mat09} has shown that there are at least $2^{2^d}$ AUSOs in 
$d$-dimensions. So just listing
their degree sequences when $d=6$ requires at least $2^{70}$ steps.
Except for extremely low dimensions, it is therefore not possible to construct all acyclic USOs.
However searching for all acyclic USOs which contain a Hamiltonian path greatly reduces the search space.
This is due to a remarkable indegree characterization due to Williamson-Hoke 
discussed in Section \ref{section:searching}.
We are able to exploit this property 'on the fly' to eliminate early prefixes of Hamiltonian paths
that cannot be completed to an acyclic USO. This is because the final indegree of each 
vertex is known as soon as it enters the path. The enumeration enabled us to see that in fact
most rules do not follow Hamiltonian paths, a fact we were then able to prove.
Of course proving that a pivot rule cannot follow a Hamiltonian
path
does not say anything about the existence or not of other exponential length paths.
However searching for these is likely to be significantly more difficult.

The paper is structured as follows. In the next section we define various
history based pivot rules that have appeared in the literature: Zadeh's original rule, least-used direction rule, least-recently considered rule, least-recently basic rule, least-recently entered rule, and least iterations in the basis rule. We also give an example that shows they are all different.
In Section \ref{section:searching} we develop
an algorithm that generates all Hamiltonian paths, if any, followed by these
history based pivot rules. We also provide computational results that show
that most of these rules do not in fact produce Hamiltonian paths for dimensions up to $7$,
except in very low dimensions.
In Section \ref{section:none} we prove this fact holds for all higher dimensions for four of the history
based rules we have presented.
\section{History based pivot rules}
\label{section:history}

In this section we review a number of history based pivot rules that have appeared in the literature starting with Zadeh's original rule. We also present an example to show that the rules all behave
differently. The difference of these rules can be seen from the difference of the history array $h$. This array is indexed
by the $2d$ directions (or sometimes, all $d$ unsigned directions), and represents current historical information
required for the given rule.

Zadeh noticed that the Klee-Minty construction (see Figure \ref{fig:klee-minty})
greatly favours some directions over others,
and designed a new pivot rule to defeat this.

\vspace{1em}
\noindent
\hangindent2em
\hangafter=0
\emph{Zadeh's rule} (a.k.a. the \emph{least entered rule}) \cite{zadeh2009worst}: 
For the entering variable, select the improving variable that has entered the basis least often thus far. (Figure \ref{fig:zadehexample})
The history array $h$ is defined on all $2d$ directions, and $h(t)$ is the number of times the direction $t$ is used. 
\vspace{1em}

\begin{figure}[!ht]
\begin{minipage}[!ht]{\textwidth}
\begin{center}
\includegraphics[width=0.5\textwidth ]{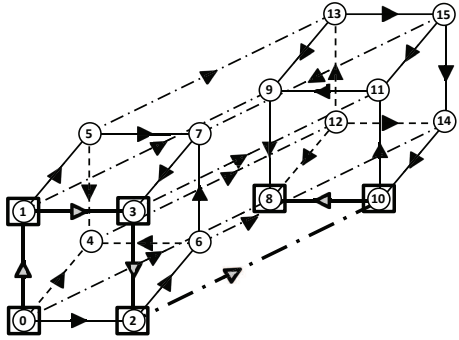}
\end{center}
\caption{Zadeh's rule}
\label{fig:zadehexample}
\end{minipage}
\def\@captype{table}
\begin{minipage}[!ht]{\textwidth}
\label{table:zadehexample}
\vspace{3em}
\begin{center}
\begin{tabular}{|c|c|c|c|c|c|c|c|c|c|}
\hline
\multicolumn{1}{|c|}{Vertex}&\multicolumn{8}{|c|}{direction}&{Outgoing direction} \\\cline{2-9}
(binary)& $+1$ & $-1$  & $+2$ & $-2$ & $+3$& $-3$ & $+4$     & $-4$  & (bold for chosen)\\\hline
$0(0000)$& $0$ & $0$  & $0$& $0$ & $0$& $0$ & $0$     & $0$  & $\boldsymbol{+1},+2,+3,+4$\\\hline
$1(0001)$& $1$ & $0$  & $0$& $0$ & $0$& $0$ & $0$     & $0$  & $\boldsymbol{+2},+3,+4$ \\\hline
$3(0011)$& $1$ & $0$  & $1$& $0$ & $0$& $0$ & $0$     & $0$  & $\boldsymbol{-1},+4$ \\\hline
$2(0010)$& $1$ & $1$  & $1$& $0$ & $0$& $0$ & $0$     & $0$  & $+3,\boldsymbol{+4}$ \\\hline
$10(1010)$& $1$ & $1$  & $1$& $0$ & $0$& $0$ & $1$     & $0$  & $+1,\boldsymbol{-2}$\\\hline
$8(1000)$& $1$ & $1$  & $1$& $1$ & $0$& $0$ & $1$     & $0$  & \\\hline
\end{tabular}
\end{center}
\end{minipage}
\end{figure}
In Zadeh's rule, as in others that we will study, there may be ties in selecting the entering variable.
We will assume that ties may be broken arbitrarily in this paper.
Note that Zadeh's rule chooses between all $2d$ variables ($d$ decision variables and $d$ slack variables) whereas the next history-based rule chooses between the $d$ pairs of decision and slack variables, $(x_i, x_{d+i})$, each of which defines a \emph{direction}.
Directions are not a very useful concept in arbitrary polytopes, as no two edges may be parallel, but they are a natural feature of hypercubes and are inherited by zonotopes,
which
are projections of hypercubes. They directly inherit the $d$ directions 
of the hypercube, some of which may no longer appear.

\vspace{1em}
\noindent
\hangindent2em
\hangafter=0
\emph{Least-used direction rule} (LUD) \cite{avis2009history}: 
For the entering variable, select the improving variable whose unsigned direction has been used least often thus far. (Figure \ref{fig:ludexample})
The history array $h$ is defined on all $d$ unsigned directions, and $h(t)$ is the number of times the direction $t$ is used. 
\vspace{1em}

\begin{figure}[!ht]
\begin{minipage}[!ht]{\textwidth}
\begin{center}
\includegraphics[width=0.5\textwidth ]{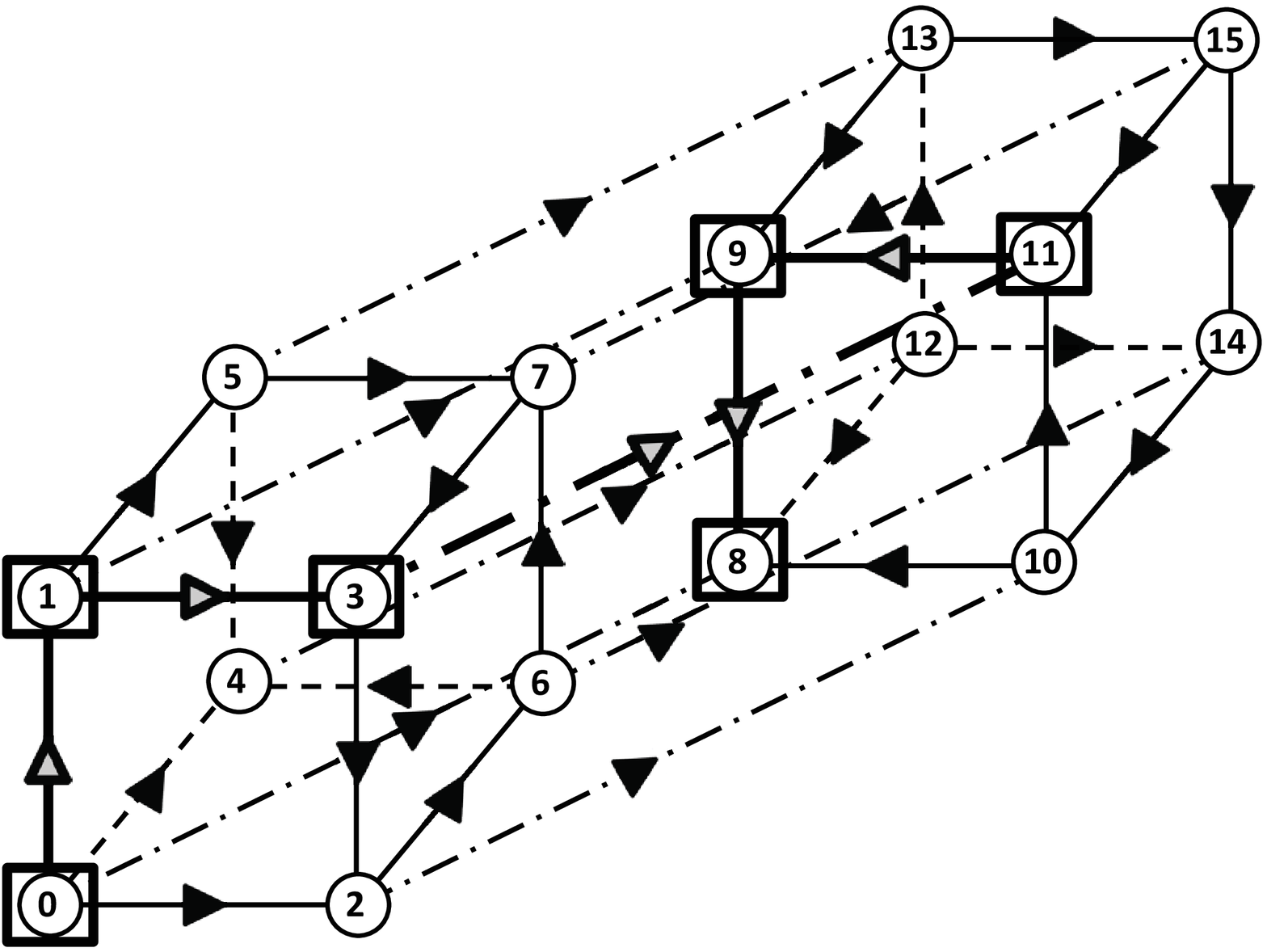}
\end{center}
\caption{Least Used Direction rule}
\label{fig:ludexample}
\end{minipage}
\def\@captype{table}
\begin{minipage}[!ht]{\textwidth}
\label{table:ludexample}
\vspace{3em}
\begin{center}
\begin{tabular}{|c|c|c|c|c|c|}
\hline
\multicolumn{1}{|c|}{Vertex}&\multicolumn{4}{|c|}{direction}&{Outgoing directions} \\\cline{2-5}
(binary)& $1$ & $2$& $3$ & $4$  & (bold for chosen)\\\hline
$0(0000)$& $0$ & $0$  & $0$& $0$ & $\boldsymbol{1},2,3,4$\\\hline
$1(0001)$& $1$ & $0$  & $0$& $0$ & $\boldsymbol{2},3,4$\\\hline
$3(0011)$& $1$ & $1$  & $0$& $0$ & $1,\boldsymbol{4}$\\\hline
$11(1011)$& $1$ & $1$  & $0$& $1$ & $\boldsymbol{2}$ \\\hline
$9(1001)$& $1$ & $2$  & $0$ & $1$ & $\boldsymbol{1}$\\\hline
$8(1000)$& $2$ & $2$  & $0$& $1$ &  \\\hline
\end{tabular}
\end{center}
\end{minipage}
\end{figure}

We now give some other history-based rules that have appeared in the literature.
We show the paths generated by these rules on the previous example in the Appendix.

\begin{itemize}
\item
\emph{Least-recently considered rule} \cite{cunningham1979theoretical}:
Fix an ordering of the variables $v_1, v_2, ...,$\\
$v_{2d}$ and let the previous entering variable be $v_i$.
For the entering variable, select the improving variable that first appears in the sequence $v_{i+1}, v_{i+2}, ...,
v_{2d}, v_1, ..., v_{i-1}$ (or $v_1, ..., v_{2d}$ if this is the first pivot).
The history array $h$ is defined on all $2d$ directions and is initialized by setting
$h(t)$ to be the rank of $t$ in the given fixed ordering. 
If direction $s$ is chosen the array is updated as 
$h(t) \leftarrow (h(t)-h(s)-1\mbox{ mod }2d) +1$. 
The Appendix shows the example of the case when initial sequence is $\{+2, -4, +1, -3, -2, +3, -1, +4\}$
(Figure \ref{fig:lrcexample})

\item
\emph{Least-recently basic rule} [Johnson in \cite{cunningham1979theoretical}]:
For the entering variable, select the improving variable that left the basis least-recently.
The history array $h$ is defined on all $2d$ directions: $h(t)$ is the step number the $|t|$-th bit of 
the vertex was last $1$ if $t$ is positive or was last $0$ if $t$ is negative.
(Figure \ref{fig:lrbexample})

\item
\emph{Least-recently entered rule} (a.k.a. \emph{least-recently used}) \cite{fathi1986affirmative}:
For the entering variable, select the improving variable that entered the basis least-recently thus far.
The history array $h$ is defined  on all $2d$ directions: $h(t)$ is the step number when the $|t|$-th 
bit of the vertex last changes from $0$ to $1$ if $t$ is positive or from $1$ to $0$
if $t$ is negative.
(Figure \ref{fig:lreexample})

\item
\emph{Least iterations in the basis rule} \cite{avis2009history}:
For the entering variable, select the improving variable that has been in the basis for the least number of iterations.
(Figure \ref{fig:lebexample})
The history array $h$ is defined on all $2d$ directions, and is the number 
of times the $|t|$-th bit of the vertex is $1$ if $t$ is positive or $0$ if $t$ is negative.
\end{itemize}

Note that all the examples illustrate distinct paths on the same AUSO cube.

In the following section we will describe an algorithm 
to determine if there are any AUSOs that admit Hamiltonian paths for the
history based methods described in this section.

\section{Searching for Hamiltonian paths on AUSOs that follow history based pivot rules}
\label{section:searching}
We developed an algorithm for determining if the various history
based rules can be made to follow a Hamiltonian path on an AUSO.
As noted in the introduction, it is known that the number of AUSOs is a doubly exponential,
so a direct search quickly becomes infeasible.
We use the fact that we are looking for AUSOs with Hamiltonian paths,
which greatly reduces the search space.
\subsection{Preliminaries}
Our basic approach is to generate Hamiltonian paths starting with an unoriented hypercube,
rather than first orienting the cube and checking if it is Hamiltonian.
Suppose that a cube has a Hamiltonian path labelled with vertices
$v_1, ... , v_N$. Then acyclicity implies immediately that each edge of the hypercube
$v_i v_j$ with $i<j$ must be directed from $v_i$ to $v_j$.
Therefore, given a Hamiltonian path on the cube, we can 
easily construct the unique acyclic orientations for all edges of the cube. 
It still remains to test whether this orientation is an AUSO. 
Fortunately there is an efficient way to
do this based on
Williamson Hoke's theorem \cite{hoke1988completely}:
\newtheorem{Hoke}{Theorem}
\begin{thm}
\label{thm:Hoke}
If an orientation on a $d$-dimensional cube is acyclic, the following conditions are equivalent.
\begin{itemize}
\item The orientation is a unique sink orientation.
\item For $k=0,...,d$ there are exactly
$d \choose k$ vertices with indegree $k$ 
(and hence $d \choose k$ vertices with outdegree $k$).
\end{itemize}
\end{thm}
This makes it very easy to check if a given Hamiltonian path appears in an AUSO cube:
we need only test the degree sequence.
Furthermore, we can even use this test as the Hamiltonian path is being constructed.
Note that when a vertex is added to the path its indegree and out-degree are known.
Also partial degree information is known for unexplored vertices.
Therefore if Williamson Hoke's condition is violated, we need not complete the
construction of the given path. This leads to an efficient pruning technique.
We also have the following interesting corollary.
\begin{col}
\label{col:indeg}
In a Hamiltonian path on a AUSO $d$-cube starting from vertex $0$, 
the indegree of a vertex is $1$ if and only it is reached by a positive direction (or, unsigned direction)
that is being used for the first time.
\end{col}
\begin{proof}
Suppose $(v,v')$ is an edge on the Hamiltonian path that uses the direction $+t$ for the first time. 
Then all previous vertices in the path must have zero on the $t$-th bit.
However all neighbours of $v'$ on the hypercube except $v$ have one on this bit, so they
cannot have been visited yet. Therefore the indegree of $v'$ is one.
Since there are $d$ directions, this yields $d$ vertices on the path with indegree one.
By Williamson Hoke's theorem this is the entire set of such vertices.
\end{proof}

As mentioned in the introduction, Zadeh's rule encourages each variable to be used
as a pivot variable roughly the same number of times. We make this precise
in the following result.

\begin{thm}
\label{thm:lowerBound}
Assume that there is a Hamiltonian path $P$ that follows Zadeh's rule on
an AUSO of an $n$-cube.
The least-used signed direction is used
at least $\frac{2^{n-2}}{n} -\frac{3}{2}$ times.
\end{thm}
\begin{proof}
We may assume that $P$ starts at vertex zero.
Let $-t$ to be the least-used direction, and $k$ be the number of times signed direction $-t$ is used.
Partition the $d$-dimensional AUSO into two $(d-1)$-dimensional hypercubes 
$C_1$ and $C_2$ where the direction $t$ separates the two (see Figure \ref{fig:twocubes})
and $P$ starts in $C_1$.

\vspace{3em}
 \begin{figure}[ht]
  \begin{center}
  \includegraphics[width=0.5\textwidth]{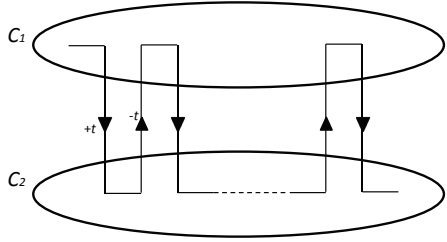}
\caption{The two $(d-1)$-dimensional cubes $C_1$ and $C_2$ separated by the direction $t$.}
\label{fig:twocubes}
\end{center}
\end{figure}

   \noindent 
Let $m_i, i=1,...,n$ be the number of times that signed direction
$+i$ is used in $P$ and $m_{n+i}, i=1,...,n$ be the number of times that signed direction
$-i$ is used in $P$.
Since all $2^n$ vertices are visited, we have
\[
\sum_{i=1}^{2n} m_i = 2^n - 1
\]
We know that the minimum value $m_{n+t}=k$ and $m_t = k+1$. 
We can estimate the sum in another way by computing $m_i$ as $P$ is followed.
Suppose we are at vertex $v$ in $C_1$ and follow a signed direction $+i$
with $m_i \geq k+2$. 
The signed direction $+t$ would have been a 
preferred choice since $m_t \leq k+1$. 
If signed direction $+t$ was not taken, then it must be that its
neighbour in $C_2$ was already visited. We call this a blocked pair.
A similar analysis holds if $v$ is in $C_2$ and a signed direction $-i$ 
is chosen with $m_{n+i} \geq k+1$.
There can be at most $2^{n-1}$ blocked pairs.
So in computing the sum of the $m_i$ along $P$ we have at most a contribution of 
$n(k+2) +n(k+1) -1$ for the unblocked pivots and a contribution of at most
$2^{n-1}$ for the blocked pivots.
Therefore
\[
\sum_{i=1}^{2n} m_i \leq n(2k+3)-1 + 2^{n-1}.
\]
Combining the two expressions for the sum, the theorem follows.
\end{proof}

Unfortunately Theorem \ref{thm:lowerBound} only holds when the path is Hamiltonian.
It is possible for a non-Hamiltonian exponential length path to use a 
signed direction as few as zero times! An example is shown in
Figure \ref{fig:nonhamiltonian}.
Here we assume that $C_1$ and $C_2$ are copies of an AUSO cube with a long path.
The resulting cube $C$ is easily seen to be an AUSO.
Note that since the path is non-Hamiltonian in $C$, vertices unvisited by the path in $C_1$ may
be directed into $C_2$.
\begin{figure}[ht]
  \begin{center}]
  \includegraphics[width=0.5\textwidth]{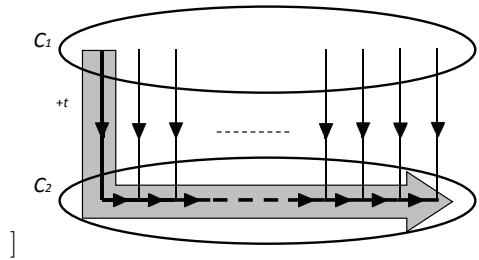}
\end{center}
\caption{An example of a non-Hamiltonian exponential path where one signed direction is never used}
\label{fig:nonhamiltonian}
\end{figure}
\subsection{The algorithm and its validity}
In this subsection we describe an algorithm that can generate, up to equivalence, all
Hamiltonian paths on AUSOs using any of the history based pivot rules described in Section \ref{section:history}. 
In this paper, when we say two paths are equivalent, it means they are equivalent up to permutation. In other words, when two paths P and Q are equivalent, there is a permutation of cordinates $f: \{1,2,\cdots,d\} \rightarrow \{1,2,\cdots,d\}$ such that $f(P)=Q$. 
Algorithm \ref{algo:search} gives the pseudocode of our algorithm.
We assume that the
Hamiltonian path starts from the vertex labelled $0$.
We denote the indegree of the vertex $x$ by $indeg(x)$.
\begin{algorithm}[ht]
\caption{Enumerate HP on AUSO-cube with history based pivot rule}
\label{algo:search}
\begin{algorithmic}[1]
\STATE $path \leftarrow \{0\}$.
\IF{current path is Hamiltonian path}
\IF{$path$ is USO}
\STATE{output the result} 
\ENDIF
\ELSE
\STATE $m \leftarrow \min_{t\in\{\pm 1,..., \pm d\}, t\mbox{ is feasible}}\{ h(t)\; | move(path.end,t) \mbox{ is not visited }\}$
\FORALL{$t$ such that $h(t) = m$}
\IF{$h(t)=0$ and  $\exists t'<t\; s.t.\; h(t') = 0$}
\STATE{continue}
\ELSE
\STATE{$v \leftarrow move(path.end,t)$}
\IF{$h(t)\neq 0$ and  $indeg(v) = 0$}
\STATE{continue}
\ENDIF
\STATE{$path \leftarrow path + v$}
\STATE{renew $h$}
\STATE{continue searching (from line $2$)}
\STATE{recover $h$}
\STATE{delete $path.end$}
\ENDIF
\ENDFOR
\ENDIF
\end{algorithmic}
\end{algorithm}
For $t=\pm 1,...,\pm d$ the function $move( x,t )$ returns the neighbour of $x$ using the signed direction $t$, that is, the vertex $x+sign(t)2^{|t|-1}$. 
Note that we focus on this function only when $t$ is a feasible move.
The array $h$ denotes the history information of the path and depends on the pivoting rule. 
For example, in the case of Zadeh's rule, $h(t)$ is the number of times the signed direction $t$ is taken.
We claim that the algorithm outputs, up to equivalence, all required Hamiltonian paths and that there are
no duplications.
First of all
we show that each of the required Hamiltonian paths are equivalent to one of the paths output by the program.
Below, by  `history based pivot rule' we refer to any of the rules described in Section \ref{section:history}.
\begin{lem}
\label{lem:completeness}
For every Hamiltonian path $P$ on a $d$-cube which can be followed by a history based pivot rule, 
there is a labelling of the cube such that $P$ begins with vertex $0$ and 
the order of positive directions first used in $P$ is $\{1,2,\dots,d-1,d\}$.
\end{lem}
\begin{proof}[Proof of Lemma \ref{lem:completeness}]
We show how to embed $P$ on a $d$-cube so that it has the required properties.
We label the first vertex in $P$ as $0$ and the initial edge of $P$ as direction $+1$.
Continuing, for $i=2,...,d$ we consider the first edge of $P$
that leaves a face of the cube of dimension $i-1$.
We define the direction used by this edge as $+i$.
This induces a labelling of the cube with the desired properties.
\end{proof}
We remark that all paths produced by Algorithm 1 satisfy the conditions of Lemma \ref{lem:completeness}
due to lines 8-10.
Next we will prove that Algorithm 1 does not produce duplicate paths.
\begin{lem}
\label{lem:nonexistence}
Let $P$ and $Q$ be two Hamiltonian paths produced by Algorithm 1.
If there is a
bijection (permutation of the cordinates) $f: \{1,2,\cdots,d\} \rightarrow \{1,2,\cdots,d\}$ such that $f(P)=Q$
then it is the identity mapping, i.e. $P=Q$. 
\end{lem}
\begin{proof}
[Proof of Lemma~\ref{lem:nonexistence}]
As remarked, both $P$ and $Q$ satisfy the conditions of Lemma \ref{lem:completeness}.
Since $f(P)=Q$, both paths must use the $k$-th positive
direction for the first time at the same time. By the lemma this must be direction $+k$,
hence $f$ is the identity mapping.
\end{proof}

As a consequence of these two lemmas we have the following result.
\begin{thm}
Algorithm 1 provides a complete duplicate free list of Hamiltonian paths on AUSO-cubes that follow a given
history based pivot rule.
\end{thm}

\subsection{Computational results}
We implemented the algorithm and ran it on an
Opteron computer with $2.2$GHz CPU, $4 \times 4 = 16$ processors and $132$GB of memory. 
We were able to do a complete enumeration up to dimension $6$ and the results are shown in
Table \ref{table:result}. The other rules refer to the least-recently considered, least-recently
basic and least iterations in the basis rules.
We see that the number of Hamiltonian path increases 
exponentially with Zadeh's least entered rule, whereas it becomes zero with the other pivot rules.
On the basis of these results we conjecture that, except for the least entered rule,
such Hamiltonian paths do not exist in any dimension greater than 6. 
We present proofs of these conjectures in the next section
for all rules except for the least-used direction rule.
\begin{table}[!ht]
\begin{center}
\caption{The number of Hamiltonian paths produced by history based pivot rules}
\label{table:result}
\end{center}
\begin{center}
\begin{tabular}{|c||c|c|c|c|c|c|}
\hline
Dimension & $2$ & $3$ & $4$ & $5$ & $6$ & $7$ \\ \hline
Least entered rule& $1$ & $2$ & $17$ & $1,072$ & $3,262,342$ & $> 10^{10}$\\ \hline
Least-used direction& $1$ & $1$ & $1$ & $2$ & $0$ & $0$  \\ \hline
Least recently entered& $1$ & $1$ & $1$ & $0$ & $0$ & $0$  \\ \hline
Least-recently considered rule & $1$ & $3$ & $13$ & $0$ & $0$ & $0$ \\ \hline
Least-recently basic rule & $1$ & $0$ & $0$ & $0$ & $0$ & $0$ \\ \hline
Least iterations in basis rule & $1$ & $0$ & $0$ & $0$ & $0$ & $0$ \\ \hline
\end{tabular}
\end{center}
\end{table}

We also conducted an experiment to check whether these paths satisfy the
Holt-Klee condition \cite{holt1998proof}, a necessary condition for realizability of an LP-cube which states that every $d$-dimensional faces have at least $d$ disjoint paths from a unique source to a unique sink (See Table \ref{table:result2}).
For dimension 7 with least entered rule, we have not found any such paths, 
but the computation was not completed due to the long running time.

\begin{table}[!ht]
\begin{center}
\caption{The number of Hamiltonian paths produced by history based pivot rules which satisfy Holt-Klee condition}
\label{table:result2}
\end{center}
\begin{center}
\begin{tabular}{|c||c|c|c|c|c|c|}
\hline
Dimension & $2$ & $3$ & $4$ & $5$ & $6$ & $7$ \\ \hline
Least entered rule& $1$ & $2$ & $12$ & $79$ & $360$ & $?$\\ \hline
Least-used direction& $1$ & $1$ & $1$ & $0$ & $0$ & $0$  \\ \hline
Least recently entered& $1$ & $1$ & $1$ & $0$ & $0$ & $0$  \\ \hline
Least-recently considered rule & $1$ & $3$ & $12$ & $0$ & $0$ & $0$ \\ \hline
Least-recently basic rule & $1$ & $0$ & $0$ & $0$ & $0$ & $0$ \\ \hline
Least iterations in basis rule & $1$ & $0$ & $0$ & $0$ & $0$ & $0$ \\ \hline
\end{tabular}
\end{center}
\end{table}

\section{Non-existence of Hamiltonian paths}
\label{section:none}
In this section we prove that, except for the least entered rule and least-used direction rule,
there are no Hamiltonian paths for the history based pivot rules considered except for those
shown in Table \ref{table:result}.
\begin{thm}
The least iterations in basis rule and the least-recently basic rule
do not have any Hamiltonian paths on a $d$-cube for $d \ge 3$.
\end{thm}
\begin{proof}
Suppose there is such a Hamiltonian path $P$ for some $d$-cube.
From Lemma \ref{lem:completeness} we can verify that the first $d$ edges of $P$
must take the directions $1,2,...,d$. Therefore $P$ begins with $d+1$ vertices
$0,1,3,...,2^d-1$. The origin has indegree $0$ and the other
$d$ vertices have indegree $1$ in the AUSO induced by $P$.
At this point, for each of the pivot rules,
the direction $-1$ has the minimal value of $h$ among the all the outgoing directions, 
so the first $d+1$ steps have to be $0,1,...,2^d-1,2^d-2$. 
The vertex $2^d-2$ also has indegree $1$, since it is not adjacent to any of the other vertices
already on $P$. There are $d+1$ vertices which have indegree $1$ in total: $1,3,7,\cdots,2^d-1,$ and $2^d-2$. 
This violates Williamson-Hoke's condition
given in Theorem \ref{thm:Hoke} which allows only $d$ vertices to have indegree $1$.
\end{proof}

We remark that this proof also can be used to show that
the least-recently considered rule cannot have a Hamiltonian path for any $d \geq 3$
if the ordering begins with $+1, +2, \cdots, +d, -1$.

\begin{thm}
\label{thm:lre}
The least-recently entered rule does 
not have any Hamiltonian paths on a $d$-cube for $d \ge 5$.
\end{thm}
\begin{proof}
Suppose $P$ is a Hamiltonian path produced by Algorithm 1 for the least-recently entered rule when $d \ge 5$.
We will show that $P$ must begin with the sequence of vertices $Q=\{Q_1,Q_2,Q_3,Q_4\}$ 
where
$Q_1=\{0,1,3,\cdots,2^d-1\}$,
$Q_2=\{2^d-1-2^{d-2},2^d-1-2^{d-2}-2^{d-3},\cdots,2^{d-1}\}$,
$Q_3=\{2^{d-1}+2,2,6,14,\cdots,2^d-2\}$ and
$Q_4=\{
2^d-2-2^{d-2},2^d-2-2^{d-2}-2^{d-3},\cdots,2^{d-1}+2+4+8,2^{d-1}+2+4\}.$

$Q$ includes the vertices $\{2^{d-1}+2,2^{d-1}+2+4+8,2^{d-1}+2+4\}$ as a subsequence 
and does not contain the vertex $2^{d-1}+2+8$. These four vertices lie on a 2-face 
which has two sources, $2^{d-1}+2$ and $2^{d-1}+2+4+8$, a contradiction (See Figure \ref{fig:contradiction}). 
\begin{figure}[!ht]
\begin{center}
\includegraphics[width=0.6\textwidth ]{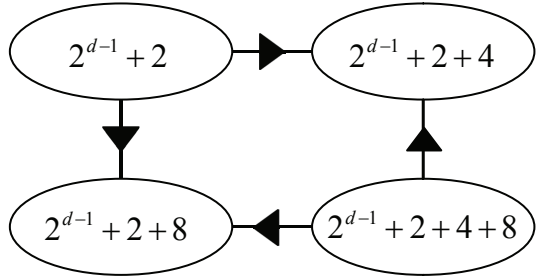}
\end{center}
\caption{$2$-dimensional face with two sources and two sinks}
\label{fig:contradiction}
\end{figure}
It remains to show that $P$ begins as specified.

\begin{itemize}
\item
$Q_1=0,1,3\cdots,2^d-1$. 
This follows from Lemma \ref{lem:completeness}.
\item
$Q_2=2^d-1,2^{d}-2^{d-2}-1,2^{d}-2^{d-2}-2^{d-3}-1,\cdots,2^{d-1}$.

We prove this by mathematical induction.
For the basic step, we will show only $2^{d}-2^{d-2}-1$ can come right after $2^{d}-1$.
When we visited the vertex $2^{d}-1$, all of the bits are $1$. 
It means the next vertex can be represented as $2^{d}-2^{k}-1 = \sum_{i=0}^{d-1}2^i-2^k\; (d-1> k \geq 0)$. By Corollary \ref{col:indeg}, vertex $\sum_{i=0}^{d-1}2^i-2^k$ should have two 
visited neighbours, one of which is obviously the vertex $2^{d}-1$.
In other words, there exists $j\neq k$ such that $\sum_{i=0}^{d-1}2^i-2^k- 2^j \in \left\{ 0,1,3\cdots,2^d-1 \right\} = \left\{v|\exists l\;s.t.\;v=\sum_{i=0}^{l} 2^i \right\}\cup\left\{0\right\}$. Since $d\geq 3$ forces $\sum_{i=0}^{d-1}2^i-2^k- 2^j$ not to be equal to $0$, $\sum_{i=0}^{d-1}2^i-2^k- 2^j$ should be represented as $\sum_{i=0}^{l} 2^i=\sum_{i=0}^{d-1} 2^i-\sum_{i=l+1}^{d-1}2^i$ for certain $l$. Therefore, the $(k,j)$ equal $(d-1,d-2)$ or $(d-2,d-1)$, and $d-1>k$ requires $k = d-2$. (See Figure \ref{fig:thm4_2} for the binary representation).

\begin{figure}[!ht]
\begin{center}
\includegraphics[width=0.6\textwidth ]{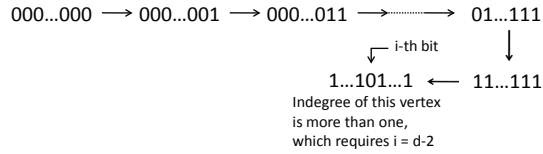}
\end{center}
\caption{Binary representation for the basic step of $Q_2$}
\label{fig:thm4_2}
\end{figure}

We can prove the inductive step similarly.
If the path is continued by $2^d-1,2^{d}-1-2^{d-2},\cdots,2^{d}-1-\left\{\sum_{i=d-2-k}^{d-2} {2^i}\right\}$, 
the next vertex should be equal to $\sum_{i=0}^{d-1}2^i-\sum_{i=d-2-k}^{d-2} {2^i}+2^j\;(d-2-k\leq j\leq d-2)$ or $\sum_{i=0}^{d-1}2^i-\sum_{i=d-2-k}^{d-2} {2^i}-2^j\;(j=d-1$ or $j < d-2-k)$.
By Corollary \ref{col:indeg}, two neighbours of it 
are in $\left\{0,1,3\cdots,2^d-1,2^{d}-1-2^{d-2},\cdots,2^{d}-1-\sum_{i=d-2-k}^{d-2} {2^i} \right\}$
Using binary numbers, $2^{d}- \left\{\sum_{k=0}^{i} {2^{d-(2+k)}}\right\} -1$ can be denoted $100\dots 0011\dots 11$, where we have $k+1$ $0$s.
 (See Figure \ref{fig:thm4_2_2} for the binary representation).
\begin{figure}[!ht]
\begin{center}
\includegraphics[width=0.6\textwidth ]{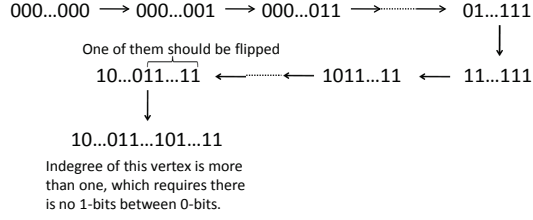}
\end{center}
\caption{Binary representation for the inductive step of $Q_2$}
\label{fig:thm4_2_2}
\end{figure}

\item
$Q_3=\{2^{d-1}+2,2,6,14,\cdots,2^d-2\}$

At the vertex $2^{d-1}$, the history array becomes 
\[
h(x) = \left\{\begin{array}{l}d+x\mbox{ (if }x > 0\mbox{ )}\\1\mbox{ (if }x=-d\mbox{ )}\\2d+1-x\mbox{ (if }d-1 \leq x\leq 0\mbox{ )}\end{array}\right.
\]
(See Table \ref{table:lrere}).
\begin{table}[!ht]
\begin{center}
\caption{The history information of least-recently entered rule. \label{table:lrere}
}
\end{center}
\begin{center}
\scalebox{0.6}{
\begin{tabular}{|c|c|c|c|c|c|c|c|c|c|c|c|c|}
\hline
\multicolumn{1}{|c|}{Vertex}&\multicolumn{11}{|c|}{direction}&\multirow{2}*{Comment} \\\cline{2-12}
(binary)& $+1$ & $-1$  & $+2$& $-2$ &$+3$& $-3$ &$\cdots$& $+d-1$& $-(d-1)$& $+d$ & $-d$  & \\\hline
$0(000...000)$& $0$ & $1$  & $0$& $1$ &$0$& $1$ &$\cdots$& $0$ & $1$  & $0$ & $1$  & initial state\\\hline
$1(000...001)$& $2$ & $1$  & $0$& $1$&$0$& $1$ &$\cdots$& $0$ & $1$  & $0$     & $1$  & \\\hline
$3(000...011)$& $2$ & $1$  & $3$& $1$ &$0$& $1$ &$\cdots$& $0$ & $1$  & $0$     & $1$  & \\\hline
$\vdots$& $\vdots$ & $\vdots$  & $\vdots$& $\vdots$ &$\vdots$&$\vdots$ &$\vdots$& $\vdots$ & $\vdots$  & $\vdots$     & $\vdots$  & \\\hline
$2^{d-1}-1(011...111)$& $2$ & $1$  & $3$& $1$ &$4$& $1$ &$\cdots$& $d$ & $1$  & $0$     & $1$  & \\\hline
$2^{d}-1(111...111)$& $2$ & $1$  & $3$& $1$ &$4$& $1$ &$\cdots$& $d$ & $1$  & $d+1$     & $1$  & end of $Q_1$\\\hline
$2^{d}-1-2^{d-2}(101...111)$& $2$ & $1$  & $3$& $1$ &$4$& $1$ &$\cdots$& $d$ & $d+2$  & $d+1$     & $1$  & \\\hline
$\vdots$& $\vdots$ & $\vdots$  & $\vdots$& $\vdots$ &$\vdots$&$\vdots$ &$\vdots$& $\vdots$ & $\vdots$  & $\vdots$     & $\vdots$  & \\\hline
$2^{d-1}+1(100...001)$& $2$ & $1$  & $3$& $2d-1$ &$4$& $2d-2$ &$\cdots$& $d$ & $d+2$  & $d+1$     & $1$  &\\\hline
$2^{d-1}(100...000)$& $2$ & $2d$  & $3$& $2d-1$ &$4$& $2d-2$ &$\cdots$& $d$ & $d+2$  & $d+1$     & $1$  &end of $Q_2$
\\\hline
$2^{d-1}+2(100...010)$& $2$ & $2d$  & $2d+1$& $2d-1$ &$4$& $2d-2$ &$\cdots$& $d$ & $d+2$  & $d+1$     & $1$  & \\\hline
$2(000...010)$& $2$ & $2d$  & $2d+1$& $2d-1$ &$4$& $2d-2$ &$\cdots$& $d$ & $d+2$  & $d+1$     & $2d+2$  & \\\hline
$6(000...110)$& $2$ & $2d$  & $2d+1$& $2d-1$ &$2d+3$& $2d-2$ &$\cdots$& $d$ & $d+2$  & $d+1$     & $2d+2$  & \\\hline
$\vdots$& $\vdots$ & $\vdots$  & $\vdots$& $\vdots$ &$\vdots$&$\vdots$ &$\vdots$& $\vdots$ & $\vdots$  & $\vdots$     & $\vdots$  & \\\hline
$2^{d-1}-2(011...110)$& $2$ & $2d$  & $2d+1$& $2d-1$ &$2d+3$& $2d-2$ &$\cdots$& $3d-1$ & $d+2$  & $d+1$     & $2d+2$  & \\\hline
$2^d-2(111...110)$& $2$ & $2d$  & $2d+1$& $2d-1$ &$2d+3$& $2d-2$ &$\cdots$& $3d-1$ & $d+2$  & $3d$     & $2d+2$  & end of $Q_3$ \\\hline

\end{tabular}
}
\end{center}
\end{table}
Although its minimum value is $1$, when $x=-d$, and the second smallest value is $2$, when $x=+1$,
we can not use either the direction $-d$ or $+1$, since they lead to visited vertices.
That leads us to use the direction $+2$, whose value is third smallest.
The vertex $2^{d-1}+2$ enables us to use the direction $-d$ at last.
Afterward, to avoid visiting an already visited vertex, we have to follow the sequence 
$\{2^{d-1}+2,2,6,14,\cdots,2^d-2\}$

\item
$Q_4=\{
2^d-2-2^{d-2},2^d-2-2^{d-2}-2^{d-3},\cdots,2^{d-1}+2+4+8,2^{d-1}+2+4\}$

This follows the same reasoning as $Q_3$, that is, using the smallest direction which reaches unvisited vertex fixes $Q_4$. Note that direction $+1$ can not be used because the destination has already been visited in $Q_2$.
\end{itemize}
\end{proof}
For the least entered rule, we can prove the following feature concerning
the beginning of any Hamiltonian path.

\begin{thm}
\label{thm:feature-zadeh}
For every $d$ dimensional Hamiltonian path using the least
entered rule, different signed directions are used for the first $2d-1$ steps.
\end{thm}

\begin{proof}Let $v$ be a vertex visited during the first $2d-1$ steps.
It is enough to show there is a direction $t$ 
and an unvisited vertex $move(v,t)$
for which $h(t) = 0$. 

If $t>0$ and $h(t ) = 0$ then vertex $move( v,t )$ cannot have been visited yet, so $t$ is a candidate direction.
Otherwise,
if for each $t>0$ we have $h( t ) = 1$ and there exists at least two negative directions 
$- t_1, -t_2$ such that $h(-t_1)=h(-t_2) = 0$.
Assume that the direction $+t_1$ was used earlier than $+t_2$.
So for all vertices visited so far, it is impossible to have both the $t_2$-th bit at $1$ and 
the $t_1$-th bit at $0$.
This means that the vertex $move(v, -t_1)$ has not been visited and 
$-t_1$ is a candidate direction.\end{proof}
\section{Discussion}
From our computational experiments, Zadeh's
least entered rule seems very likely to have Hamiltonian paths on AUSO cubes.
Using our program, we could verify such paths exist up to dimension 9, but did not
yet find any for dimension 10.
Furthermore, we could not find any general construction, so this is an open problem.
Even if such Hamiltonian paths exist, it is not clear whether or not they could be obtained on
AUSOs that are realizable as polytopes.

Although we showed that a number of history based pivot rules do not admit 
Hamiltonian paths in general, they may still admit exponential length paths.
Since our program makes heavy use of the fact that we are searching for Hamiltonian
paths, we were not able to use it to check this for low dimensions.

Our computer result for Zadeh's rule allow ties to be broken arbitrarily, as does the theoretical lower bound obtained in \cite{Fr10}.
It would be interesting to see the effects of various deterministic tie breaking rules on these results.
\section{Acknowledgments}
We are grateful to ERATO-SORST Quantum Computation and Information
Project, Japan Science and Technology Agency.
All of our computational experiments are conducted on the cluster
computer in ERATO-SORST.
Work on this project was also supported by an INTRIQ-ERATO/SORST collaboration
grant funded by MDEIE(Qu\'{e}bec), a discovery grant
from NSERC(Canada), and KAKENHI(Japan).

We thank the anonymous reviewers for their useful comments to improve this paper. 
\bibliographystyle{plain}
\bibliography{simplex_method}

\newpage
\section{Appendix}
\begin{figure}[!ht]
\begin{minipage}[!ht]{0.5\textwidth}
\begin{center}
\includegraphics[width=\textwidth ]{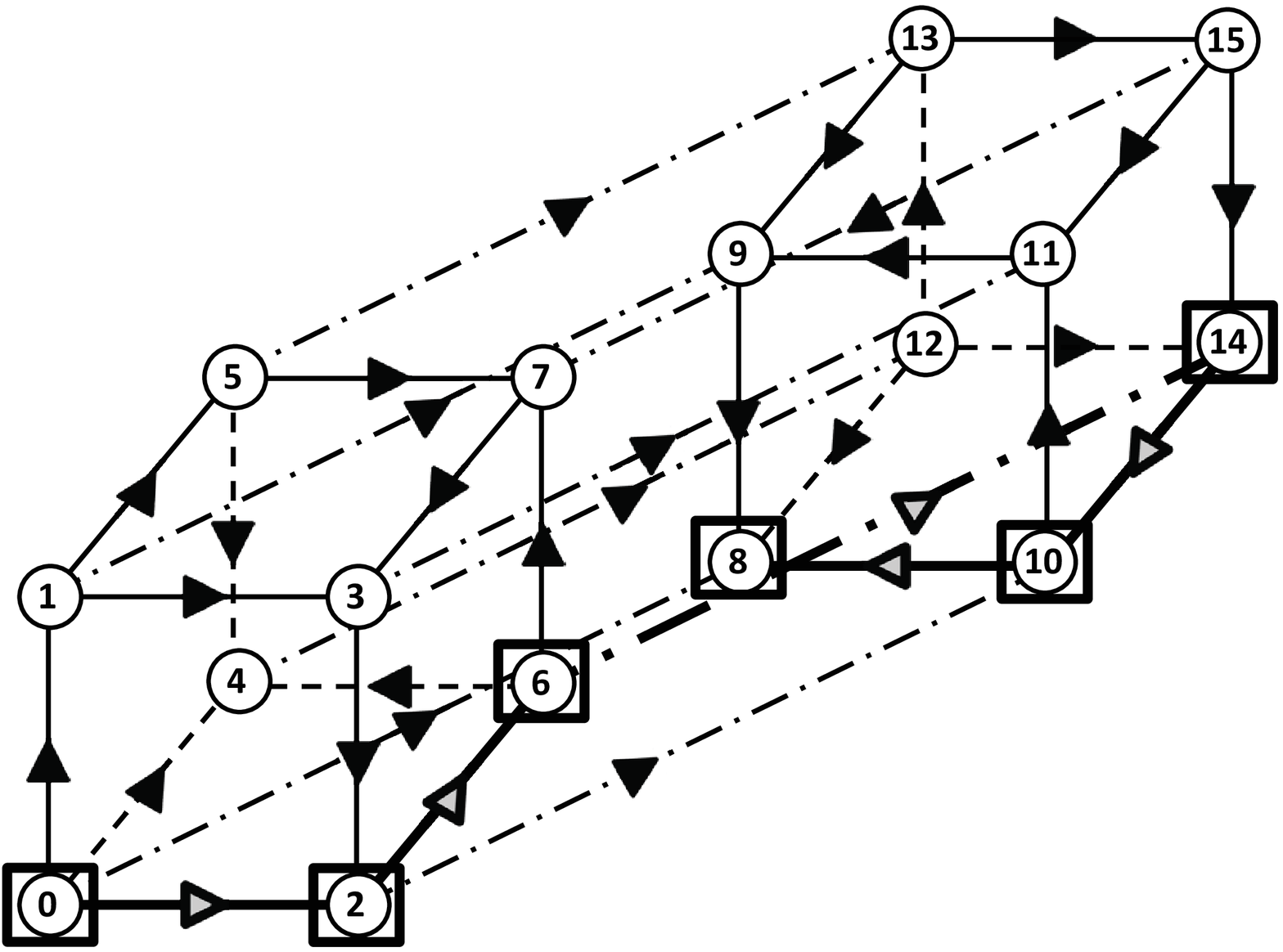}
\end{center}
\caption{Least recently considered rule}
\label{fig:lrcexample}

\end{minipage}
\def\@captype{table}
\scalebox{0.65}[1.3]{
\begin{minipage}[!ht]{0.5\textwidth}
\label{table:lrcexample}
\begin{center}
\begin{tabular}{|c|c|c|c|c|c|c|c|c|c|}
\hline
\multicolumn{1}{|c|}{Vertex}&\multicolumn{8}{|c|}{direction}&{Outgoing directions} \\\cline{2-9}
(binary)& $+1$ & $-1$  & $+2$& $-2$ & $+3$& $-3$ & $+4$     & $-4$  &(bold for chosen) \\\hline
$0(0000)$& $3$ & $7$  & $1$& $5$ & $6$& $4$ & $8$     & $2$  & 
$+1,\boldsymbol{+2},+3,+4$\\\hline
$2(0010)$& $2$ & $6$  & $8$& $4$ & $5$& $3$ & $7$     & $1$  & 
$\boldsymbol{+3},+4$\\\hline
$6(0110)$& $5$ & $1$  & $3$& $7$ & $8$& $6$ & $2$     & $4$  & 
$+1,-2,\boldsymbol{+4}$\\\hline
$14(1110)$& $3$ & $7$  & $1$& $5$ & $6$& $4$ & $8$     & $2$  & 
$\boldsymbol{-3}$\\\hline
$10(1010)$& $7$ & $3$  & $5$& $1$ & $2$& $8$ & $4$     & $6$  & 
$+1,\boldsymbol{-2}$\\\hline
$8(1000)$& $6$ & $2$  & $4$& $8$ & $1$& $7$ & $3$     & $5$  & \\\hline
\end{tabular}

\end{center}
\end{minipage}
}
\end{figure}

\begin{figure}[!h]
\begin{minipage}[!h]{0.5\textwidth}
\begin{center}
\includegraphics[width=\textwidth ]{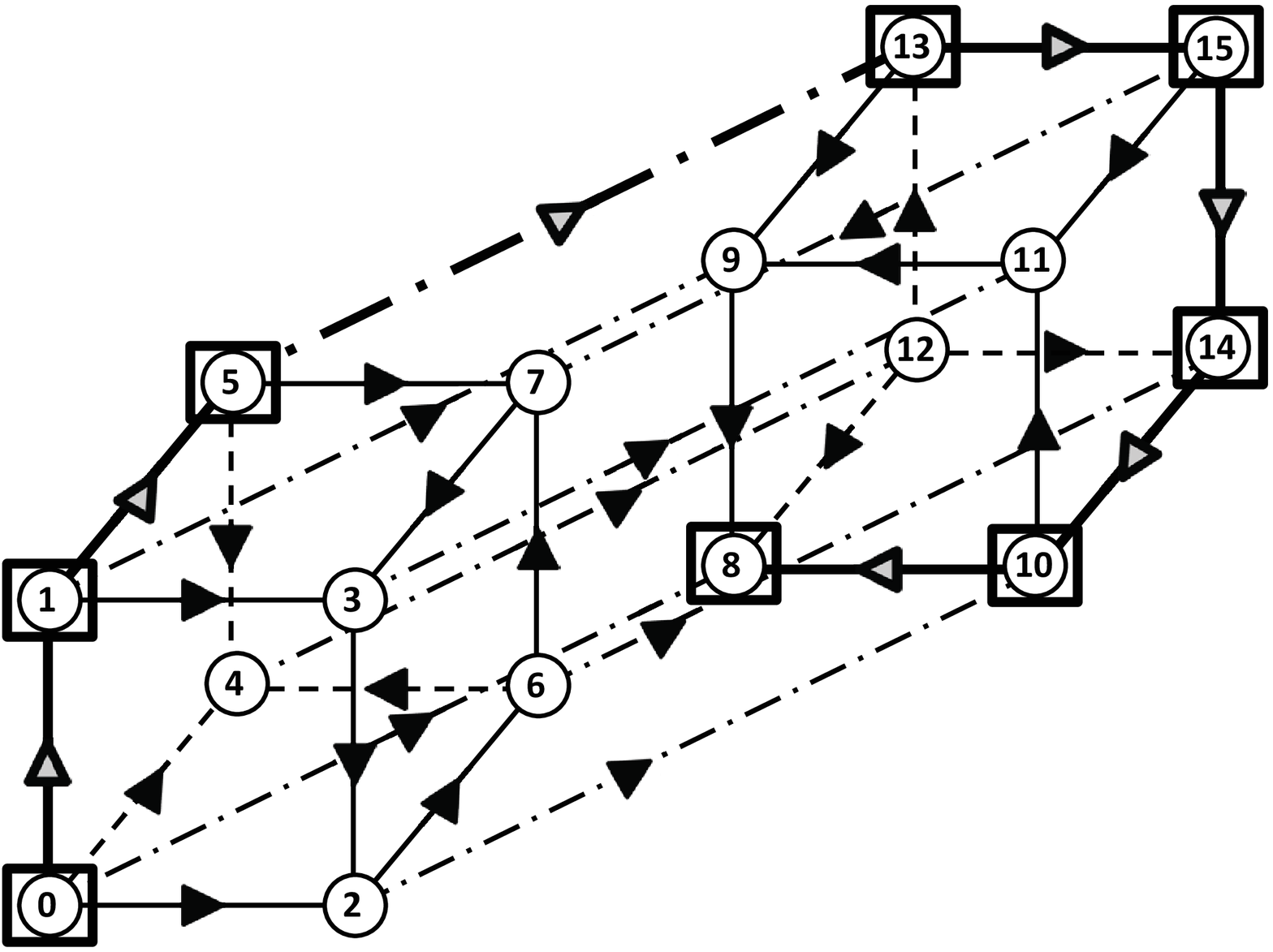}
\end{center}
\caption{Least recently basic rule}
\label{fig:lrbexample}
\end{minipage}
\def\@captype{table}
\scalebox{0.5}[1.3]{
\begin{minipage}[!h]{\textwidth}
\label{table:lrbexample}
\begin{center}
\begin{tabular}{|c|c|c|c|c|c|c|c|c|c|c|}
\hline
\multicolumn{1}{|c|}{Step}&\multicolumn{1}{|c|}{Vertex}&\multicolumn{8}{|c|}{direction}&Outgoing directions \\\cline{3-10}
number&(binary)& $+1$ & $-1$  & $+2$& $-2$ & $+3$& $-3$ & $+4$     & $-4$  & (bold for chosen)\\\hline
1&$0(0000)$& $0$ & $1$  & $0$& $1$ & $0$& $1$ & $0$     & $1$  & $\boldsymbol{+1},+2,+3,+4$\\\hline
2&$1(0001)$& $2$ & $1$  & $0$& $2$ & $0$& $2$ & $0$     & $2$  & $+2,\boldsymbol{+3},+4$\\\hline
3&$5(0101)$& $3$& $1$  & $0$& $3$ & $3$& $2$ & $0$     & $3$  & $-1,+2,\boldsymbol{+4}$\\\hline
4&$13(1101)$& $4$ & $1$  & $0$& $4$ & $4$& $2$ & $4$     & $3$  & $\boldsymbol{+2},-3$\\\hline
5&$15(1111)$& $5$ & $1$  & $5$& $4$ & $5$& $2$ & $5$     & $3$  & $\boldsymbol{-1},-3,-4$\\\hline
6&$14(1110)$& $5$ & $6$  & $6$& $4$ & $6$& $2$ & $6$     & $3$  & $\boldsymbol{-3}$\\\hline
7&$10(1010)$& $5$ & $7$  & $7$& $4$ & $6$& $7$ & $7$     & $3$  & $+1,\boldsymbol{-2}$\\\hline
8&$8(1000)$& $5$ & $8$  & $7$& $8$ & $6$& $8$ & $8$     & $3$  & \\\hline
\end{tabular}
\end{center}
\end{minipage}
}
\end{figure}
\begin{figure}[!t]
\begin{minipage}[!t]{0.5\textwidth}
\begin{center}
\includegraphics[width=\textwidth ]{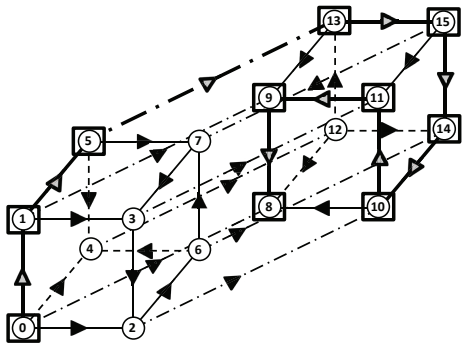}\end{center}
\caption{Least recently entered rule}
\label{fig:lreexample}
\end{minipage}
\def\@captype{table}
\scalebox{0.5}[1.3]{
\begin{minipage}[!t]{\textwidth}
\label{table:lreexample}
\begin{center}
\begin{tabular}{|c|c|c|c|c|c|c|c|c|c|c|}
\hline
\multicolumn{1}{|c|}{Step}&\multicolumn{1}{|c|}{Vertex}&\multicolumn{8}{|c|}{direction}&Outgoing directions \\\cline{3-10}
number&(binary)& $+1$ & $-1$  & $+2$& $-2$ & $+3$& $-3$ & $+4$     & $-4$  & (bold for chosen)\\\hline
1&$0(0000)$& $0$ & $1$  & $0$& $1$ & $0$& $1$ & $0$     & $1$  & $\boldsymbol{+1},+2,+3,+4$\\\hline
2&$1(0001)$& $2$ & $1$  & $0$& $1$ & $0$& $1$ & $0$     & $1$  & $+2,\boldsymbol{+3},+4$\\\hline
3&$5(0101)$& $2$ & $1$  & $0$& $1$ & $3$& $1$ & $0$     & $1$  & $-1,+2,\boldsymbol{+4}$\\\hline
4&$13(1101)$& $2$ & $1$  & $0$& $1$ & $3$& $1$ & $4$     & $1$  & $\boldsymbol{+2},-3$\\\hline
5&$15(1111)$& $2$ & $1$  & $5$& $1$ & $3$& $1$ & $4$     & $1$  & $-1,\boldsymbol{-3},-4$\\\hline
6&$11(1011)$& $2$&  $1$  & $5$& $1$ & $3$& $6$ & $4$     & $1$  & $\boldsymbol{-2}$\\\hline
7&$9(1001)$& $2$ & $1$  & $5$& $7$ & $3$& $6$ & $4$     & $1$  & $\boldsymbol{-1}$\\\hline
8&$8(1000)$& $2$ & $8$  & $5$& $7$ & $3$& $6$ & $4$     & $1$  & \\\hline
\end{tabular}
\end{center}
\end{minipage}
}
\end{figure}
\begin{figure}[!ht]
\begin{minipage}[t]{0.5\textwidth}
\begin{center}
\includegraphics[width=\textwidth ]{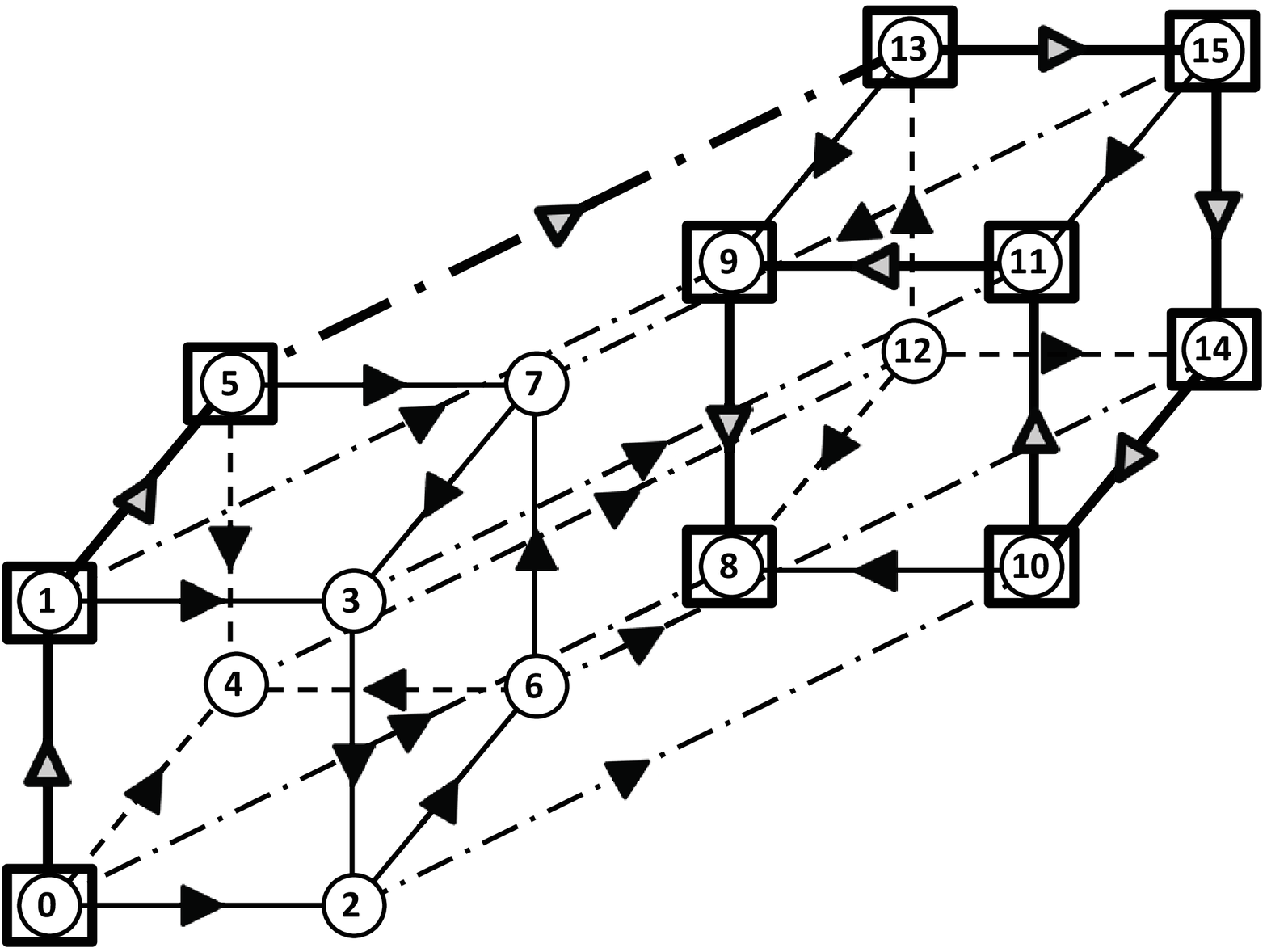}
\end{center}
\caption{Least iterations in basis rule}
\label{fig:lebexample}
\end{minipage}
\hfill
\def\@captype{table}
\scalebox{0.5}[1.3]{
\begin{minipage}[!ht]{\textwidth}
\label{table:lebexample}
\begin{center}
\begin{tabular}{|c|c|c|c|c|c|c|c|c|c|}
\hline
\multicolumn{1}{|c|}{Vertex}&\multicolumn{8}{|c|}{direction}&{Outgoing directions} \\\cline{2-9}
(binary)& $+1$ & $-1$  & $+2$& $-2$ & $+3$& $-3$ & $+4$     & $-4$  &(bold for chosen) \\\hline
$0(0000)$& $0$ & $1$  & $0$& $1$ & $0$& $1$ & $0$     & $1$  & $\boldsymbol{+1},+2,+3,+4$\\\hline
$1(0001)$& $1$ & $1$  & $0$& $2$ & $0$& $2$ & $0$     & $2$  & $+2,\boldsymbol{+3},+4$\\\hline
$5(0101)$& $2$ & $1$  & $0$& $3$ & $1$& $2$ & $0$     & $3$  & $-1,+2,\boldsymbol{+4}$ \\\hline
$13(1101)$& $3$ & $1$  & $0$& $4$ & $2$& $2$ & $1$     & $3$  & $\boldsymbol{+2},-3$\\\hline
$15(1111)$& $4$ & $1$  & $1$& $4$ & $3$& $2$ & $2$     & $3$  & $\boldsymbol{-1},-3,-4$\\\hline
$14(1110)$& $4$ & $2$  & $2$& $4$ & $4$& $2$ & $3$     & $3$  & $\boldsymbol{-3}$\\\hline
$10(1010)$& $4$ & $3$  & $3$& $4$ & $4$& $3$ & $4$     & $3$  & $\boldsymbol{+1},-2$\\\hline
$11(1011)$& $5$ & $3$  & $4$& $4$ & $4$& $4$ & $5$    & $3$  & $\boldsymbol{-2}$\\\hline
$9(1001)$& $5$ & $3$  & $4$& $5$ & $4$& $5$ & $6$     & $3$  & $\boldsymbol{-1}$\\\hline
$8(1000)$& $5$ & $4$  & $4$& $6$ & $4$& $6$ & $7$     & $3$  & \\\hline
\end{tabular}
\end{center}
\end{minipage}
}
\end{figure}

\end{document}